\documentclass{article}
\usepackage[utf8]{inputenc}
\usepackage{amsmath}
\usepackage{amsthm}
\usepackage{amsfonts}
\usepackage{graphicx}
\usepackage{graphics}
\usepackage{array}

\newtheorem{example}{Example}

\newtheorem{assum}{Assumption}
\newtheorem{proposition}{Proposition}

\newtheorem{lemma}{Lemma}

\begin{document}

\title{Incentive Schemes for Rollup Validators}

\author{
Akaki Mamageishvili and Edward W. Felten\\ Offchain Labs, Inc.
}

\maketitle

\begin{abstract}
    We design and analyze attention games that incentivize validators to check computation results. We show that no pure strategy Nash equilibrium of the game without outside parties exists by a simple argument. We then proceed to calculate the security of the system in the mixed Nash equilibrium, as a function of the number of validators and their stake sizes. Our results provide lower and upper bounds on the optimal number of validators. More concretely, a minimal feasible number of validators minimizes the probability of failure. The framework also allows to calculate optimum stake sizes, depending on the target function. In the end, we discuss optimal design of rewards by the protocol for validators.
\end{abstract}

\section{Introduction}

Decentralized systems, such as blockchains, rely on validators to compute and verify the state of the system. Validators are rational players that provide a service when the right incentives are in place: they are rewarded for correct behavior and punished for misbehavior. Blockchains have successfully implemented these incentives. The validators stake some amount of (blockchain native) tokens and earn rewards if they provide validation service. If they misbehave, there are two ways they get punished. First, directly, by slashing stakes. Second, indirectly: if the system loses security, it causes the validator stake to be devalued or even stolen.

There are still many research directions to explore. We study a particular setting in this paper, specifically, a setting where a single honest and active validator is necessary for the security of a system, as for example described in~\cite{arbitrum}. We ask a  question exactly how many validators are required for a high enough security level of the system, and what their stake sizes should be. Our model and insights are general enough to be applied to a broad class of decentralized systems, but they are best suited to optimistic rollup protocols. To scale the base layer protocol, computationally heavy execution is delegated from the base layer to layer two protocols, called rollup protocols. In case of optimistic rollups, only claims on intermediate states are made on the base layer. This is where active validation becomes important. If one party, typically one of the validators, wrongly claims a state of the system, other validators are supposed to challenge and prove that the claim is wrong. A single honest validator is enough to disprove a wrong claim, and therefore, guarantee the security of the protocol.

In the first part of the paper, we study a game where validators obtain rewards only when they find false claims made by other (malicious) validators. In the context of optimistic rollups, only one honest validator which checks computation results is enough to detect and prevent malicious attackers' efforts. 
The malicious asserter can make a false claim that allows all assets at the rollup to be stolen if all validators fail to challenge it. This is why long periods for a challenge are given. Successful attacks are subtler on the base blockchains, as they also devalue stolen assets if they are in the native tokens, however, in the rollups the deployed assets are not native and therefore, not devalued on the base layer.  

We first consider a game in which the behavior of players is completely based on economic incentives. In particular, there are no social norms and the payoffs of the players are not affected by outside parties, e.g., the system designer does not inject money into the system. The reward a validator might receive is paid by another validator that made a false claim about the system state. This way the game is similar to the so-called multi-player {\it zero-sum game}. The difference with the zero-sum game, however, is that if the malicious validator succeeds with a false claim, it may obtain a much higher payoff than what the rest of the validators lose.   

In this model, we argue that a pure strategy equilibrium does not exist, and therefore, we focus on mixed strategy equilibrium solution concept. In this solution, we calculate the probability of the system failure, which should serve as an upper bound on the system security expected if there are only rational players. In reality, we expect the system to be more secure because there are altruistic players and some players may also care about their reputation.

In the second part, we allow rewards to validators without detecting false claims. To achieve this, validators are required to post their checking results on-chain with some probability. To be able to figure out if posting is needed, validators need to compute the state of the system fully. If the probability of posting is high enough, validators are incentivized to do computation all the time, hence, pure strategy equilibrium is obtained. They are rewarded when they correctly post about the state of the world, the amount which compensates for their past checking efforts as well. 

\subsection*{Related Literature}
The question of correctly incentivizing validators in blockchain settings is not recent, see~\cite{dem_incentives}, where the problem is referred to as the {\it verifier's dilemma}. 
In parallel research,~\cite{li2023security} studies the security of optimistic rollups in a very similar model as ours. The differences are the absence of deposits of the validators, equivalent to our definition of a silent validator and focusing on 2 players, a representative attacker, an {\it Asserter} in our case, and a representative validator.
A closely related paper to ours is~\cite{DAC}. The difference between our approaches is that in ~\cite{DAC} validators (data providers) are not explicitly rewarded for their service. Instead, the authors look at only punishing strategies to incentivize the right behavior. Also, they consider pure strategy Nash equilibria solutions.~\cite{finding_bug} studies a similar question of how to design a bug-searching committee. The paper studies the question of the optimal number of searchers and reward policies, given a fixed budget.~\cite{reward_sharing} studies blockchain security through the lenses of game theory, in particular, how to design a sharing of validation rewards. However, in this model, the costs of validation differ across validators. 
Another related paper is~\cite{staking_pools}, where costs of validation differ, and the authors look at the problem of delegating validation, a common way of earning through rewards. 
Game theory and security literature have a big intersection. For research relating these two in the context of multi-party computation settings as well, see~\cite{GT_sec_comp}.

\section{Model}

We first consider the simplest case of two validators, which we sometimes refer to as players. One of them is playing the role of the Asserter and the other playing the role of the Checker. This setting corresponds to the system with two validators. Most of the parameters are defined in the game with two players, denoted by $\mathcal{G}_2$ and later used for a more general setting with $n$ players.  
The roles of players are specified in the following: 

\begin{itemize}
    \item The Asserter makes a claim about the state of a system that is either true or false. The state of the system is a result of a computation of a state transition function, that is fed a stream of incoming transactions.
    \item The Checker can check the claim made by the Asserter, at some computational cost. The Checker has the option to challenge the claim.
\end{itemize}

The Asserter's action set consists of false and true claims.
The Checker's action set consists of checking and not checking. 

The following is a list of the parameters of the game that define the final utilities of the players.

\begin{itemize}
    \item  $C$: the cost of checking. Here we refer to the cost of checking each claim about the state of the system. Claims arrive at a regular rate. We ignore the sunk costs of setting up the validator node. After such a node is set up, the validator that runs a node pays only server and maintenance costs. Node software that does a check of the claim is typically available for validators of the rollup systems. Therefore, it is safe to assume that they incur homogenous costs of checking. 
    \item  $R$: the deposit of the Asserter. In case of false claim discovery, $R$ is also the Checker's reward. In practice, only a fraction of the reward is given to a Checker and the rest is burnt. This is done in order to avoid situations where both validators are controlled by the same malicious player. In such a case this malicious player does not lose any deposit by rising a false assertion and losing the challenge to itself~\footnote{Such delay is an attack on its own.}. However, in the analysis, we assume that validators are not controlled by the same player and the mechanism is budget balanced: what one player loses goes to another player. The analysis does not change and results are the same asymptotically (and qualitatively) if we assume burning some fraction of the Asserter's deposit.
    \item  $L$: the Checker's loss if cheating goes undetected. Note that $L$ does not have to be equal to $R$. It can be a stake validators have locked for being a validator. However, by the design of a challenge mechanism, $L$ is at most $R$. 
    \item  $U$: is the Asserter's gain if cheating goes undetected. It can be thought of as the full assets locked in the protocol, sometimes referred to as {\it total value locked (TVL)}. This interpretation of the value is most suited to rollup protocols. 
\end{itemize}

\begin{center}
\begin{tabular}{ | m{4cm} | m{2cm}| m{2cm} | } 
  \hline
  strategy & false ($\pi$) & true ($1-\pi$) \\ 
  \hline
  check ($\alpha$) & $R-C,-R$ & $-C, 0$ \\ 
  \hline
  don't check ($1-\alpha$) & $-L$, $U$ & $0$,$0$ \\ 
  \hline
\end{tabular}
\end{center}

The table above gives a game $\mathcal{G}_2$ in a bimatrix format. The first number is the utility (payoff) of the Checker, the second number is the utility of the Asserter.
The Checker is sometimes referred to as a row player and the Asserter -- a column player.

We make a mild assumption on the parameters: 
\begin{assum}
Assume that the Checker does not have a dominant strategy, that is, $R-C>-L$, equivalent to $C<R+L$. 
\end{assum}

The condition is intuitive: if the cost of checking is too high, the Checker never checks. A common goal is to minimize  the probability of a system failure, denoted by $F(C,L,R,U)$. That is, the probability that a false claim is introduced by the Asserter and the Checker does not check it.
Note that the Asserter does not have a dominant strategy, given the assumption. Since the Checker does not have a dominant strategy, we can only have a totally mixed equilibrium game. A similar observation is made in~\cite{li2023security}. However, in the totally mixed equilibrium solution of~\cite{li2023security}, in our case, because of the assumptions, the probability of a successful attack is very low, but the damage is huge, while in the case of~\cite{li2023security}, both probability and damage are moderate. We assume that if an attack is detected, then the Asserter gets nothing, while in~\cite{li2023security} gives the same $\pi$ payoff in both cases. In the case of detection, the payoff is negative because of a larger deposit slashed, while in our case it is just negated deposit, $-R$.
Second, we assume that in case of a successful attack, the attacker can in principle steal all assets, while~\cite{li2023security} assumes that the $\pi$ is a moderate value. 

We compute probabilities in the mixed equilibrium solution in the following. 
Consider mixed strategies for both players. For the Asserter, strategies can be characterized by the probability $\pi$ that the chosen action (claim) is false. With probability $1-\pi$, the claim is true. 
For the Checker, the mixed strategy is characterized by the probability $\alpha$ that the Checker checks. With probability $1-\alpha$, the Checker does not check. 

\begin{proposition}\label{dep_cost_tvl}
The probability of failure is increasing in the cost of checking $C$ and decreases in $U$.    
\end{proposition}

\begin{proof}
 We can now calculate equilibrium probabilities $\alpha$ and $\pi$, using indifference conditions. The indifference condition for the Checker is that the expected utility of playing ''check" is equal to the expected utility of playing ''don't check". That is,

\begin{equation}\label{indifference_checker_2_players}
\pi(R-C)+(1-\pi)(-C)=\pi(-L),    
\end{equation}

equivalent to 

\begin{equation}\label{prob_false_claim}
    \pi=\frac{C}{R+L}.
\end{equation} 

The indifference condition for the Asserter is that the expected utility of playing ''false" is equal to the expected utility of playing ''true". That is, 

\begin{equation}\label{indifference_asserter_2players}
    0\alpha + 0(1-\alpha)=(-R)\alpha + U(1-\alpha),
\end{equation} 
equivalent to 

\begin{equation}\label{alpha_formula}
    \alpha=\frac{U}{R+U}.
\end{equation} 

By plugging in the equilibrium values of $\pi$ and $\alpha$, we obtain:

\begin{equation}\label{failure_prob}
    F(C,L,R,U):=\pi(1-\alpha)=\frac{CR}{(R+L)(R+U)}.
\end{equation} 

It is easy to see that $F(C,L,R,U)$ is decreasing in increasing $U$.

\end{proof}

If the rollup has a higher value, it is less likely to fail in the equilibrium. The explanation is simple: a higher value of the rollup protocol makes it more attractive for a malicious Asserter to try to make a false claim to transfer all value to itself, but this on the other hand gives more motivation to the Checker to check, as it is earning on finding a false claim. In light of this, another value of interest is the expected loss of the system $F(C,L,R,U)U$. This value is increasing in $U$, and converges to $\frac{CR}{R+L}$ as $U$ tends to infinity. 

Note that, by~\eqref{alpha_formula}, $\alpha$ is decreasing in increasing $R$. This sounds counterintuitive -- a larger reward for the checker discovering a false claim makes it less likely that the checker will check for a false claim. 
Even though higher $R$ should increase the incentive of the Checker to check, and all else being equal it does, it also decreases the probability $\pi$ that the Asserter introduces a false claim, which on its own decreases the incentive for the Checker to check. 
That is, the recommendation is that increasing $R$ is not the solution to maximize checking probability. Consider a derivative of $F(C,L,R,U)$ as a function of $R$.

\begin{equation}\label{deriv_sec}
    \frac{dF}{dR}=\frac{C(R+L)(R+U)-CR(R+L+R+U)}{(R+L)^2(R+U)^2}=\frac{CUL-CR^2}{(R+U)^2(R+L)^2}.
\end{equation}

Solving $\frac{dF}{dR}=0$ gives $R=\sqrt{UL}$. Therefore, when $R<R^*:=\sqrt{UL}$, LHS of~\eqref{deriv_sec} is positive and when $R>R^*$, then it is negative. That is, the probability is decreasing in $R$ above $R^*$, and should be taken as high as possible.

Note a few observations on $F$, from the formula~\eqref{failure_prob}. First, $F$ is minimized at $R=0$. It implies that a false claim comes for free, therefore, the Asserter tries a false claim all the time, and the Checker checks all the time, as the cost of checking is less than the punishment for not checking: $C<L$. This would be a desirable solution, but each false claim delays  finality, and therefore, harms the system. Therefore, in optimizing the parameter sets, we also care about $\pi$, which is maximized by taking $R=0$. One immediate takeaway from the~\eqref{prob_false_claim} formula is that if we want to minimize $\pi$, we need to increase $L$ and $R$.

For decreasing $\pi$, we need to increase $L$, equivalent to disincentivizing the Checker to stay idle, causing the Asserter to introduce a false claim less often, and to decrease $R$, equivalent to incentivizing the Asserter to introduce the false claim more often and, therefore, causing the Checker to check more often.

Note that $\alpha=1$ for any $R$ is achievable if the Asserter commits to introduce the false claim with a certain minimum probability: $\pi > \frac{C}{R+L}$. However, this can not be sustained in equilibrium: the Checker always checks as it has positive utility, while the Asserter has strictly negative utility: $-\pi\cdot R$. It is not rational for the Asserter. It can only be supported as a solution if, for example, the protocol designer plays the role of the Asserter and posts wrong claims with a probability more than the bound $\frac{C}{R+L}$.

One potential goal a system designer can have is to optimize social welfare, in which the costs of validating and a fraction of the stakes are subtracted from the success probability times TVL. The first cost is obvious - the cost of (duplicate) checking is lost for the validator. The second, opportunity cost, is incurred by the validators by staking their assets in the validation system instead of earning interest outside. 
The system designer has to minimize the following target function: 

\begin{equation}\label{target}
    M:=fU\pi + U\pi(1-\alpha)+\alpha C + r(L+R), 
\end{equation}

where $f$ denotes the relative loss of the system when there is a delay and $r$ is a potential return on investment outside the system. $f$ is typically assumed to be a low number, say $10^{-2}$. Plugging in the values for $\pi$ and $\alpha$ gives an equivalent equation to~\eqref{target}: 

\begin{equation*}\label{target_new}
    M=f\frac{C}{R+L} + \frac{C}{R+L}\frac{R}{R+U} = \frac{C}{R+L}\left(f+\frac{R}{R+U}\right) + \frac{U}{U+R}C+r(L+R). 
\end{equation*}

The optimum values of $R$ and $L$ minimizing $M$ can be computed by taking partial derivatives of $M$ with respect to $R$ and $L$. 

\subsection{Extension to $n+1$ validators}

In this section, we assume that there are $n+1$ validators. One of the validators is an Asserter, in each round. We want to incentivize the validators to check claims often enough. In each round, the Asserter makes a claim, and validators can check it. If they check and find the false claim, they do not get slashed if they post a challenge to the false claim. If they check and find out that the claim is true, they do not need to post anything.
That is, not posting anything can mean two things: the validator checked and found out the claim is true, or the validator did not check and that is why there is no post. 
There are two ways to implement the payoff to the players in the protocol: 

\begin{enumerate}
    \item Simultaneous: every validator posts, if they want to post, at the end of a predefined time interval. This approach is simpler to analyze, and for the players, it is simpler to make a decision. 
    \item Sequential: The validators see what other validators have done so far. If nobody posts anything this may motivate them not to post anything, but that increases the chances that someone will post a fraud-proof at the last second and slash the silent validators.
\end{enumerate} 

For simplicity, we focus on the simultaneous model in this paper. That is, validators see only at the end of the round how many posted checks. 
Having homogenous costs of checking among validators is a natural assumption in the setting of rollups, as there is available software for running a validator node and standard hardware requirements.

If no validator detects a false claim, the Asserter proceeds with the false claim, and all validators are punished by losing all their deposited stake -- $L$ -- and the Asserter can steal all value on the chain, giving it payoff $U$. 

We consider a fully mixed symmetric equilibrium of this game. 
Similarly to the case with two players, the probability that each validator checks is denoted by $\alpha$, and the probability the Asserter claims a false claim is denoted by $\pi$.
The timeline of the events is the following: 

\begin{itemize}
    \item If $m$ out of $n$ validators find a false claim and post about it, they are paid equally: $\frac{R}{m}$.
    \item the other $n-m$ validators are slashed $s_w$, which we assume to be (much) smaller than $L$. 
\end{itemize}

The probability that at least one out of $n$ validators will check is equal to $$P_{s,\alpha}(n):=1-(1-\alpha)^n.$$  Note that in this definition $\alpha$ is an independent parameter, however, in the equilibrium it depends on $n$. 
Similarly to~\eqref{indifference_checker_2_players}, we derive the indifference condition for the validator. Shortly, it is EU[check] = EU[don't check], where EU[X] stands for expected utility from taking an action $X$. The condition can be translated as:

\begin{align}\label{indifference_checker_n_players}
    & \pi\left(\sum_{i=0}^{n-1}{n-1\choose i}\alpha^i(1-\alpha)^{n-1-i}\frac{R}{i+1}\right) + (1-\pi)0 - C = \\
    &\pi(P_{s,\alpha}(n-1)(-s_w) + (1-P_{s,\alpha}(n-1))(-L)) + (1-\pi) 0.
\end{align}

The first summand on the left-hand side (LHS) represents the product of the probability that the claim is false, the probability that $i$ other validators check and (expected) rewards $R/(i+1)$, as there are $i+1$ validators finding the false claim.

The second summand on the LHS is a product of the probability that the claim is true with $0$, while the last represents the minus cost of checking. 

The first summand of the right-hand side (RHS) represents the product of the probability that the claim is false, the product of the probability that someone else checks with $-s_w$.

The second summand of RHS is a product of the probability that nobody checks with $-L$, while the third summand is a product of the probability that the claim is true with $0$.

The indifference condition can be further simplified to: 

\begin{align}\label{indifference_checker_n_players_simplified}
    &\pi\left(\sum_{i=0}^{n-1}{n-1 \choose i}\alpha^i(1-\alpha)^{n-1-i}\frac{R}{i+1}\right) - C = \nonumber \\
    &\pi(P_{s,\alpha}(n-1)(-s_w) + (1-P_{s,\alpha}(n-1))(-L)).
\end{align}

First, the following claim, obtained in~\cite{finding_bug}, holds. 

\begin{lemma} \label{lemma:modbinom}
For $n \in \mathbb{N}$ and $x \neq 0$,
\begin{equation}\sum_{k=0}^n {n \choose k} \frac{x^ky^{n-k}}{k+1} = \frac{1}{n+1}\frac{(x+y)^{n+1}-y^{n+1}}{x}.
\end{equation}
\end{lemma}


The lemma implies simplification: 

\begin{equation}\label{indif_simplified}
    \sum_{i=0}^{n-1}{n-1 \choose i}\alpha^i(1-\alpha)^{n-1-i}\frac{R}{i+1} = R\frac{1-(1-\alpha)^n}{n\alpha}.
\end{equation}

Similarly to~\eqref{indifference_asserter_2players}, the indifference condition of the Asserter is EU[false claim] = EU[true claim]. 
The LHS is the sum of the product of the probability that some validator checks with $-R$ and the product of the probability that no validator checks with $U$. The RHS is equal to $0$. 
The condition simplifies to: 

\begin{equation}\label{indifference_asserter_nplayers}
    (1-(1-\alpha)^n)R = (1-\alpha)^nU.
\end{equation}

From the condition, we obtain a solution for $\alpha$: $$\alpha = 1-\frac{R}{R+U}^{1/n}.$$

RHS of~\eqref{indif_simplified} further simplifies to: 

\begin{equation*}
R\frac{1-(1-\alpha)^n}{n\alpha} =\frac{UR}{(R+U)n\alpha}.    
\end{equation*}

Plugging in $\alpha$ in the equation~\eqref{indifference_checker_n_players_simplified} and taking into account~\eqref{indif_simplified} derives $\pi$: 

\begin{equation}
    \pi = \frac{C}{\frac{UR}{(U+R)n\alpha} + P_{s,\alpha}(n-1)s_w + (1-P_{s,\alpha}(n-1))L}.
\end{equation}

The main value of interest as in the case with 2 players is $P_{s,\alpha}(n)$. It is obtained by solving the Asserter’s indifference condition~\eqref{indifference_asserter_nplayers} and is equal to $\frac{U}{R+U}$. That is, it does not depend on the number of validators $n$. The second value of interest is $\pi$, as $\pi(1-P_{s,\alpha}(n))$ is the probability that the false claim will go through unnoticed. Note that $\pi$ is decreasing in increasing $s_w$ and $L$. That is, for decreasing the probability that the Asserter is introducing the false claim, and therefore, decreasing the total probability the false claim goes through unnoticed, we need to increase the slashed stakes of the validators. All else being equal, we obtain the following result: 

\begin{proposition}
   $\pi$ is increasing with increasing $n$.
\end{proposition}

This result helps to find out the optimum number of validators, in particular, it suggests to rollup systems that $n$ should be as low as possible. On the other hand, $n>1$ might still be needed because, for example, some validators are not online for some time.

Suppose $t$ validators go offline for a technical reason. We calculate the probability that the system still functions in equilibrium. It is equal to $P_{s,\alpha}(n-t)$. In the following, we give a numerical example.

\begin{example}
In this example, we consider realistic values of parameters. Suppose $C$ is normalized to $1$ (dollar), which is a reasonable approximation of one round computation costs, $U=10^9$ corresponding to the TVL, $R=10^6$ corresponding to a stake that an asserting validator needs to commit and $L=10^5$ corresponding to a stake an active validator needs to commit. 
The probability of failure is minimized when $n=1$ and it is equal to $\frac{10^6}{(10^6+10^5)(10^9+10^6)}\approx 10^{-9}$. Now assume that $n=12$. Then the probability each validator checks, $\alpha$, is approximately equal to $0.4377$, and the probability of a false claim is approximately equal to $3.448\cdot 10^{-6}$ and the probability of failure equals to $3.445\cdot 10^{-6}$. When $t=2$, the probability of failure equals to $0.343\cdot 10^{-6}$. The driving force of these (good) results is that $\pi$ is very low. The other multiplier, the probability that one of the validators will check, has a lower effect on the result. The following table shows approximate values of $\alpha$ and $\pi$ for $n\leq 12$. 
Note that the probability all validators will fail to catch a false claim, in that case, is $(1-\alpha)^n = \frac{R}{U+R}$, independent of $n$, and approximately equals to $10^{-3}$.

\begin{centering}
\begin{table}
\begin{tabular}{ |c|c|c|c|c| c| c| c| c| c| c| c| c|}
\hline
 $n$ & $1$ & $2$ & $3$ & $4$ & $5$ & $6$ & $7$ & $8$ & $9$ & $10$ & $11$ & $12$ \\ 
 \hline
 $\alpha$ & $0.999$ &  $0.968$ & $0.900$ & $0.822$ & $0.748$ & $0.683$ & $0.627$ & $0.578$ & $0.535$ & $0.498$ & $0.466$ & $0.437$ \\  
 \hline
 $\pi$   & 9.1e-7 & 1.9e-6 & 2.7e-6 & 3.3e-6 & 3.7e-6 & 4.1e-6 & 4.4e-6 & 4.6e-6 & 4.8e-6 & 5e-6 & 5.1e-6 & 5.3e-6 \\
\hline
\end{tabular}
\caption{Equilibrium probabilities as a function of $n$.}
\end{table}
\end{centering}

\end{example}

\subsection{Silent validators}

In this section, we assume the existence of ''silent" validators. They do not stake anything, unlike (active) validators considered so far, but can access the base layer contract after each claim and challenge the (false) claim of the Asserter, in case active validators did not do so~\footnote{The role of a silent validator can be played by the Asserter as well. That is, instead of stealing all the assets in the system, it may only collect $nL$ and allow the system to survive.}. 
A successful claim by a silent validator allows it to collect all stakes -- $nL$ and the Asserter's deposit $R$. This gives more incentive to the staked validators to check. 

To get an intuition, we start with the smallest instance. Assume that there is one active and one silent validator. The indifference condition of the active validator stays the same as in the case without silent validators, as the active validator loses its deposit if the Asserter's claim is false. The indifference condition of the silent validator, on the other hand, is:

\begin{equation}\label{silent_2_validators}
\pi (1-\alpha)R=C.   
\end{equation}

The expected gains for the silent validator is $(1-\alpha)R$, as the active validator finds the false claim with probability $\alpha$.

Plugging in $\pi=\frac{C}{R+L}$ in~\eqref{silent_2_validators} gives a contradiction, the LHS is always lower than the RHS. This implies that the silent validator never checks in the equilibrium. The same holds even if we add the active validator's deposit to the reward of a silent validator. The indifference condition in this case becomes: 

\begin{equation}\label{silent_2_validators_extra}
\pi (1-\alpha)(R+L)=C.   
\end{equation}

However, the mechanics change when we consider more than $1$ active validator. 
Suppose there are $m$ silent validators. 
The indifference condition of such a validator is different from the active validator, as it does not risk losing stake if it does not check. On the other hand, if the silent validator checks while no staked validator does, it will be rewarded both by staked validator stakes -- $nL$ -- and the dishonest Asserter's stake $R$. Silent validators have the same cost of checking, $C$. By a similar argument as with only active validators, it is easy to show that there is no pure Nash equilibrium solution to the game. 
The proof is by contradiction: if one of the validator types checked with  certainty, it would make malicious Asserter not make a false claim, causing validators not to check. Therefore, we again consider a totally mixed Nash equilibrium solution. The probability that the silent validator plays the checking strategy is $\beta$. Then, the indifference condition for the silent validator is:  

\begin{equation}\label{indif_silent}
    C = \pi(1-(1-\alpha)^n)\sum_{i=0}^{m-1}\left({m-1 \choose i}\beta^i(1-\beta^{m-1-i})\frac{R+nL}{i+1}\right). 
\end{equation}

The indifference condition for active validator is the same as~\eqref{indifference_checker_n_players_simplified}, as for this type of validator it does not matter what silent validators will do. If there is a false claim and active validators do not find it, they will lose all their stakes. For completeness, we state the condition here: 

\begin{equation}\label{indif_active}
    \pi\left(\sum_{i=0}^{n-1}{n-1 \choose i}\alpha^i(1-\alpha)^{n-1-i}\frac{R}{i+1}\right) - C = \pi(P_{s,\alpha}(n-1)(-s_w) + (1-P_{s,\alpha}(n-1))(-L)).
\end{equation}

The indifference condition for the malicious Asserter is: 

\begin{equation}\label{indif_asserter}
   (1-(1-\alpha)^n)(1-(1-\beta)^m)R=((1-\alpha)^n+(1-\beta)^m-(1-\alpha)^n(1-\beta)^m)U. 
\end{equation}

Analyzing indifference conditions~\eqref{indif_silent},~\eqref{indif_active} and~\eqref{indif_asserter} gives conditions on the parameters when totally mixed equilibrium of the game exists. Consider, for example, $m=2$ and $n=1$. The indifference conditions become: 

\begin{equation*}
C=\pi \alpha (R+2L), C=\pi (R+L), (2\alpha - \alpha^2)\beta R = (1-(2\alpha - \alpha^2)\beta)U. 
\end{equation*}

This solves $\alpha = \frac{R+L}{R+2L}$, $\beta = \frac{U}{(2\alpha-\alpha^2)(R+U)}$ and $\pi=\frac{C}{R+L}$. That is, $\alpha$ needs to be high enough, to make sure that $\beta$ is smaller than $1$. When $\alpha$ is low enough, then it must be that $\beta=1$. That is, silent validators always check.

\section{Protocol level incentives}

In this section, we ask the question of how to reward validators for checking (and posting about) the true claim~\footnote{For a similar discussion for Ethereum validator incentivization see~\text{https://dankradfeist.de/ethereum/2021/09/30/proofs-of-custody.html}.}. The post, a transaction to a smart contract at the base layer network does not need to include proof. This approach allows obtaining a pure strategy equilibrium, in which all validators check with certainty. This guarantees that a false claim is found with probability one. However, it comes at the cost of adding new functionality to the protocol, which is usually not desirable.  

Similar to the previous section, there are $n+1$ validators. The probability that each one needs to post about checking the state of the world is denoted by $P$. Validators need to take a decision whether to check or not before they find out whether they need to post about the result. In case they fail to post when they are required to post, they have slashed their stakes $L$. The cost of checking is $C$, as before. The cost of posting is $c$, which is typically assumed to be less than $C$. The opportunity cost of staking on the platform is $rL$ in each round. Therefore, $r$ denotes the return on investment outside the system in one round. The payment validators receive for posting the right outcome is denoted by $p$. Then, the expected payoff is equal to $-C+P(-c+p)$ when the validator checks, or $P(-L)$ when the validator does not check. To guarantee that the validator checks in equilibrium, the expected payoff of checking needs to be larger than the expected payoff of not checking: $$-C+P(-c+p)>P(-L).$$ This gives a condition on $P$, namely $P > \frac{C}{p-c+L}=:\pi_l$. 

The expected budget of the protocol per round is equal to $nP p$, which is lower bounded by $n\pi_lp=\frac{nCp}{p-c+L}$. Note that taking a high enough $L$ lowers the expected cost of the system to guarantee {\it incentive compatibility (IC)}, but it increases the cost to guarantee {\it individual rationality (IR)}. The latter means that validators want to be a part of the system in the first place, instead of staying away and obtaining zero utility. To guarantee IR, we need to offset the opportunity cost the validator incurs by staking $L$, which equals $rL$. Since by IC, the validator always checks, we need that $-C+P(-c +p)>rL$, that is, $P>\frac{C+rL}{p-c}=:\pi_r$. This simplifies to the condition that $\pi$ needs to be larger than $\max(\pi_c,\pi_r)$. The minimum value of $\pi$ is achieved at the minimax. One function is decreasing in $L$, another is increasing. The minimax is achieved when they are equal. That is, $$\frac{C+rL}{p-c} = \frac{C}{p-c+L},$$ implying $L = \frac{rc-rp-C}{r}$ (which is negative, therefore, not possible) or $L=0$.  This optimization is done for fixed $p$. We can minimize $P$ and $P p$ over $p$ as well. Plugging in $L=0$ into the formula of $P$ gives: $P=\frac{C}{p-c}$. It is minimized for $p$ as large as possible. Similarly, $pP=\frac{Cp}{p-c}$ is minimized for $p$ as large as possible. The value approaches $C$, which is intuitive: 

\begin{enumerate}
    \item the system pays exactly the cost of checking on average,
    \item it checks the validators with very low probability,
    \item when they are checked - the system pays a very large amount $p$.
\end{enumerate} 

Unless there is some cost associated with high payment $p$ for the protocol to upper bound it, this is an optimal strategy. However, such costs are obvious. The protocol cannot invest an arbitrarily high amount at once in rewarding validators. 

\subsection*{Implementation}

We present the implementation of attaching a message of checking and posting with some probability to an assertion. The sampling can be done on the protocol level, by referring to state-relevant hash values. 
Suppose the Asserter is making a claim about the value of $f(x)$ for some function $f$ which is common knowledge, and a value $x$ which varies across different runs of the protocol.  We want to pose a randomly generated challenge to the validator, such that the checker must know $f(x)$ in order to respond correctly to the challenge. Then we can punish the validator for responding incorrectly.

The validator has a private key $k$, with a corresponding public key $g^k$ which is common knowledge ($g$ is a suitable generator of a group where the Diffie-Hellman problem is hard).
To issue a challenge for the computation of $f(x)$, the asserter generates a random value $r$, then publishes $(x, g^r)$ as a challenge.
A validator who has private key $k$ should respond to the challenge by posting a tiny transaction on-chain if and only if $H(g^{rk}, f(x)) < T$, where $H$ is a suitable hash function and $T$ is a suitably chosen threshold value.

Note that only the Asserter (who knows $r$) and the validator (who knows $k$) will be able to compute the hash because they are the only two parties who can compute $g^{rk}$. Note also that computing the hash requires knowledge of $f(x)$.
After the validator(s) have had a window of time to post their response(s) to the challenge, the Asserter can post its claimed $f(x)$ which will be subject to challenge if any validator disagrees with it. 

At this time, the Asserter can accuse any validator who responded incorrectly: the Asserter must publish $r$ to substantiate its accusation. If the Asserter's claimed value of $f(x)$ is later confirmed, a smart contract can verify the accusation and punish the misbehaving validator (if the Asserter's claimed $f(x)$ is rejected, the Asserter's accusation is ignored).

If any funds are seized from validators, the Asserter gets half of the seized funds and the remainder is burned.
One way to build this in the rollup is to have assertions, rather than revealing the state root $f(x)$. Instead, include an attention challenge $(x, g^r, H(x, g^r, f(x)))$, which is also a binding commitment to $f(x)$, and only reveal $f(x)$ when there is a challenge, or when the assertion is confirmed.  
Validators could self-identify and stake, and they would have until the confirmation time of the assertion to post their response to the attention challenge. 
 
\section{Conclusions and future work} 

We initiate a study of the optimal number of validators and their stake sizes in the rollup protocols. The main result is that for higher system security guarantees, the cost of checking should be low, TVL should be high and the number of validators should be as low as possible. We also derive optimal validation and assertion deposits in the equilibrium. 
Future avenues of research include weighted staking. Even if such staking is not allowed, if one validator creates multiple identities, but checks only once, it results in weighted staking. Such validator's indifference condition is different from the others, as it has invested $kL$ tokens, where $k$ is the number of copies it created. 

\bibliography{sample}
\bibliographystyle{plain}

\end{document}